\documentclass[preprint,amsmath]{revtex4}
\usepackage{amssymb}
\usepackage{graphicx,epsfig,subfigure,dsfont,amssymb,amsthm,amsfonts,amsbsy,mathrsfs,amscd,appendix}
\def\qed{\leavevmode\unskip\penalty9999 \hbox{}\nobreak\hfill
     \quad\hbox{\leavevmode  \hbox to.77778em{%
               \hfil\vrule   \vbox to.675em%
               {\hrule width.6em\vfil\hrule}\vrule\hfil}}
     \par\vskip3pt}
\usepackage{tikz}  
\usepackage{pgf}   

\usepackage{graphicx}
\usepackage{subfigure}

\newtheorem{theorem}{Theorem}
\newtheorem{corollary}{Corollary}

\input amssym.def

\begin{document}

\title{The Average Quantum Coherence of Pure State Decomposition}

\smallskip
\author{Ming-Jing Zhao$^1$}
\author{Teng Ma$^{2}$ }
\author{Rajesh Pereira$^3$ }

\affiliation{
$^1$School of Science,
Beijing Information Science and Technology University, Beijing, 100192, China\\
$^2$ Beijing Academy of Quantum Information Sciences, Beijing 100193, China\\
$^3$ Department of Mathematics and Statistics, University of Guelph, Guelph, N1G2W1, Canada
}

\begin{abstract}
We study the {average} quantum coherence over the pure state decompositions of a mixed quantum state.
An upper bound of the average quantum coherence is provided and sufficient conditions for the saturation of the upper bound are shown. These sufficient conditions always hold for two and three dimensional systems. This provides a tool to estimate the average coherence experimentally by measuring only the diagonal elements, which remarkably requires less measurements compared with state tomography.  We then describe the pure state decompositions of qubit state in Bloch sphere geometrically. For any given qubit state, the optimal pure state decomposition achieving the maximal average quantum coherence as well as three other pure state decompositions are shown in the Bloch sphere.
The order relations among their average quantum coherence are invariant for any coherence measure.
The results presented in this paper are universal and suitable for all coherence measures.
\end{abstract}

\maketitle

\section{Introduction}

Quantum coherence, {an important feature in quantum world,} is a valuable resource in many areas of quantum computation and quantum communication processing, such as  quantum algorithms, quantum metrology, quantum channel discrimination, quantum thermodynamics, etc. \cite{A. Streltsov-rev,M. Hu,I. Marvian,M. Hillery}. In this context, the formulation of quantifying quantum coherence is initiated and some coherence measures such as the $l_1$ norm of coherence \cite{T. Baumgratz}, the relative entropy of coherence \cite{T. Baumgratz}, intrinsic randomness of coherence \cite{X. Yuan}, coherence concurrence \cite{X. Qi}, distillable coherence \cite{A. Winter}, robustness of coherence \cite{C. Napoli}, geometric coherence \cite{A. Streltsov}, coherence number \cite{S. Chin}, are proposed from different aspects to characterize quantum coherence. These coherence measures are not separate but connected with each other quantitatively \cite{S. Rana16, S. Rana17}.

For pure states, the quantification of coherence is well understood and the restrictions of all coherence measures to pure states are
identical to some real symmetric concave functions mathematically \cite{S. Du,H. Zhu}. However, for mixed states which have infinitely different pure state decompositions, the quantification is more complicated and relatively difficult to study. Technically most quantifiers focus on the density matrices of quantum states. But, in fact, the quantum ensembles store more information than  is available solely from the density matrix.
In this paper we shall study the quantum coherence from the aspect of the average coherence with respect to the pure state decompositions to reveal the features of the coherence in mixed states.

Operationally, the average quantum coherence can be interpreted in the following way.
Suppose Alice holds a state $\rho^A$ with coherence $C(\rho^A)$. Bob holds another part of the purified state of $\rho^A$. The joint state between Alice and Bob is $\sum_k p_k |\psi_k\rangle_A \otimes |k\rangle_B$ with $\rho^A=\sum_k p_k |\psi_k\rangle\langle \psi_k|$. Bob performs local measurements $\{|k\rangle\}$ and informs Alice of the measurement outcomes by classical communication. Alice's quantum state will be in a pure state ensemble $\{ p_k,\  |\psi_k\rangle \}$ with average coherence $\sum_k p_k C(|\psi_k\rangle\langle \psi_k|)$. In this process the coherence in Alice can be increased from $C(\rho^A)$ to the average coherence $\sum_k p_k C(|\psi_k\rangle\langle \psi_k|)$ because of the convexity of the coherence measure. This is called the assisted coherence distillation. For the average coherence, its maximum
which quantifies a one way coherence
distillation rate is called the coherence of assistance \cite{E. Chitambar,zhao l1max} and its minimum is identified as a coherence measure \cite{H. Zhu,S. Du}.

Since the average coherence depends on the pure state decompositions of mixed states, the order relation of the average coherence with respect to different pure state decompositions may change with different coherence measures generally. Let us consider quantum state $\rho=\left(
\begin{array}{ccccccc}
2/3 & 1/3\\
1/3 & 1/3
\end{array}
\right)$ with two pure state decompositions. One pure state decomposition is $\mathfrak{D}_1=\{p_i, |\psi_i\rangle \}_{i=1}^2$, $|\psi_1\rangle=\frac{1}{\sqrt{2}}(|0\rangle+|1\rangle)$ with probability $p_1=\frac{2}{3}$, $|\psi_2\rangle=|0\rangle$ with probability $p_2=\frac{1}{3}$. The second pure state decomposition is the spectral decomposition $\mathfrak{D}_2=\{p_i^\prime, |\psi_i^\prime \rangle \}_{i=1}^2$, $|\psi_1^\prime\rangle=(-1 + \sqrt{5})/\sqrt{10 - 2 \sqrt{5}}|0\rangle+ \sqrt{2/(5 - \sqrt{5})}|1\rangle$ with probability $p_1^\prime=\frac{3 + \sqrt{5}}{6}$, $|\psi_2^\prime\rangle=-(1 + \sqrt{5})/\sqrt{10 + 2\sqrt{5})}|0\rangle+ \sqrt{2/(5 + \sqrt{5})}|1\rangle$ with probability $p_2^\prime=\frac{3 - \sqrt{5}}{6}$.
For these two pure state decompositions $\mathfrak{D}_1$ and $\mathfrak{D}_2$, the first average coherence is larger than the second for the relative entropy of coherence measure while the reverse is true for the $l_1$ norm of coherence measure. This inconsistency is expected because
each coherence measure just describes one aspect of the coherence.
Surprisingly, in a single qubit system, we find the order relation to be invariant of the average quantum coherence for some pure state decompositions.
The order relation of the average quantum coherence remains the same for all coherence measures. The average quantum coherence and the corresponding pure state decompositions are intrinsic features hidden in the pure state decompositions of mixed states.

In this paper, we study the average quantum coherence over the pure state decomposition for mixed quantum state. In sec.\ II, an upper bound for the average quantum coherence is derived.  This upper bound is only a function of the diagonal entries of the quantum state. Two sufficient conditions for the saturation of the upper bound are shown which hold for two and three dimensional systems. One surprising result is the pure state decomposition attaining the upper bound is deterministic and independent of the explicit coherence measure.
In sec.\ III, we study the pure state decompositions of a qubit state in Bloch sphere geometrically. For any given qubit state, the optimal pure state decomposition reaching the maximal average quantum coherence, the spectral decomposition as well as other two pure state decompositions are shown in the Bloch sphere. Their average quantum coherence are compared and the order relations are invariant for all coherence measures. In sec.\ IV, we conclude these results with a summary.

\section{The upper bound of the average quantum coherence}

In this section, we shall evaluate the average quantum coherence over pure state decompositions for a given mixed state.
First let $\Omega=\{(x_1, x_2, \cdots, x_n)^T| \sum_i x_i=1,\ x_i\geq 0\}$ be the probability simplex. A real-valued function $f$ is called a concave function on the probability simplex $\Omega$, if
\begin{equation}\label{eq concave}
f(\lambda \textbf{x}_1+(1-\lambda)\textbf{x}_2)) \geq \lambda f( \textbf{x}_1)+(1-\lambda) f(\textbf{x}_2)
\end{equation}
for any $\textbf{x}_1$ and $\textbf{x}_2$ in $\Omega$ and any $0\leq \lambda \leq 1$. Especially, $f$ is strictly concave if inequality (\ref{eq concave}) is strict inequality
whenever $\textbf{x}_1\neq \textbf{x}_2$ and $0< \lambda < 1$.
Let ${\mathfrak{{F}}}_{{sc}}=\{f\}$ be the set of real functions on the probability simplex satisfying the following three conditions: (i) $f((1,0,\cdots,0)^T)=0$; (ii) $f$ is symmetric, which means $f$ is invariant under any permutation transformation, $f( \textbf{x})=f(P_{\pi} \textbf{x})$ with $P_{\pi}$ any permutation matrix. (iii) $f$ is concave.

For any $f\in {\mathfrak{{F}}}_{{sc}}\setminus \{0\}$, we can obtain a coherence measure for pure state $|\psi\rangle=\sum_{i=0}^{n-1} \psi_i |i\rangle$ under the reference basis $\{|i\rangle\}_{i=0}^{n-1}$ as follows:
\begin{equation}\label{cf}
C_f(|\psi\rangle)=f(|\psi_0|^2, |\psi_1|^2, \cdots, |\psi_{n-1}|^2).
\end{equation}
The vector $(|\psi_0|^2, |\psi_1|^2, \cdots, |\psi_{n-1}|^2)^T$ is called the coherence vector of a pure state $|\psi\rangle$ \cite{S. Du,H. Zhu}. Conversely,
the restriction of any coherence measure $C$ to pure states is
identical to $C_f$ for certain $f\in {\mathfrak{{F}}}_{{sc}}\setminus \{0\}$ \cite{S. Du,H. Zhu}.

For any coherence measure $C$ and any given mixed state $\rho$, the average quantum coherence with respect to the pure state decomposition $\mathfrak{D}=\{p_k, |\psi_k\rangle\}$ of $\rho$ is
\begin{equation}
{\bar C}(\mathfrak{D})=\sum_k p_k {C_f}(|\psi_k\rangle).
\end{equation}
Generally, the average quantum coherence of $\rho$ changes with different pure state decompositions.

\begin{theorem}\label{th upper bound}
For any mixed quantum state $\rho=\sum_{i,j} \rho_{ij} |i\rangle\langle j|$, its average quantum coherence is bounded from above by ${\bar C} \leq f(\rho_{11}, \rho_{22}, \cdots, \rho_{nn})$. The equality holds if there is a pure state decomposition $\{p_k,\  |\psi_k\rangle\}$ of $\rho$ such that $\Delta(|\psi_k\rangle\langle\psi_k|)=\Delta(\rho)$ for all $k$. Here $\Delta(\rho)$ denotes the diagonal matrix having the same diagonal entries as $\rho$.
\end{theorem}

\begin{proof}
Suppose $\mathfrak{D}=\{p_k,\  |\psi_k\rangle\}$ is an arbitrary pure state decomposition of $\rho$ with $|\psi_k\rangle=\sum_i \psi_i^{(k)} |i\rangle$,
then $\sum_k p_k|\psi_i^{(k)}|^2=\rho_{ii}$ for all $k$ and $i$.
So
${\bar C}(\mathfrak{D})=\sum_k p_k C_f(|\psi_k\rangle)= \sum_k p_k f(|\psi_0^{(k)}|^2, |\psi_1^{(k)}|^2,\cdots, |\psi_{n-1}^{(k)}|^2)\leq f(\sum_k p_k|\psi_0^{(k)}|^2, \sum_k p_k|\psi_1^{(k)}|^2,\cdots, \sum_k p_k|\psi_{n-1}^{(k)}|^2)=f(\rho_{11}, \rho_{22}, \cdots, \rho_{nn})$, where the inequality is because of the concavity of the function $f$. The equality holds if the vectors $(|\psi_0^{(k)}|^2, |\psi_1^{(k)}|^2,\cdots, |\psi_{n-1}^{(k)}|^2)^T$ are the same for all $k$.
\end{proof}

Theorem \ref{th upper bound} provides an upper bound for the average quantum coherence in terms of the diagonal entries of quantum state. Since $\rho_{ii}$ is the probability corresponding to the outcome of the projective measurement $|i\rangle\langle i|$, so in order to derive this upper bound experimentally, one just needs to perform the projective measurements $\{|i\rangle\langle i|\}_{i=0}^{n-1}$. The required number of the measurements for estimating the average coherence is much less than for full quantum state tomography.

To verify the saturation of the upper bound in  Theorem \ref{th upper bound}, we have employed the dephasing operation $\Delta(\rho)=\sum_i |i\rangle\langle i| \rho |i\rangle\langle i|$. In fact, we also can use correlation matrices (positive semidefinite matrices with all diagonal entries equal to one \cite{CV}) to examine the saturation of the upper bound.
For quantum state $\rho$, let $CM(\rho)=\Delta(\rho)^{-\frac{1}{2}}\rho \Delta(\rho)^{-\frac{1}{2}}$, then $CM(\rho)$ is a correlation matrix. So whether the average coherence of $\rho$ attains the upper bound $f(\rho_{11}, \rho_{22}, \cdots, \rho_{nn})$ is equivalent to whether $CM(\rho)$ can be decomposed as the convex combination of rank one correlation matrices.

\begin{theorem}\label{th saturation}
The average coherence ${\bar C}$ reaches the upper bound $f(\rho_{11}, \rho_{22}, \cdots, \rho_{nn})$ if $CM(\rho)$ is the convex combination of rank one correlation matrices.
\end{theorem}

\begin{proof}
If $CM(\rho)$ is the convex combination of rank one correlation matrices, $CM(\rho)=\sum_k p_k A_k$ with all $A_k$ rank one, $A_k=|v_k\rangle\langle v_k|$, with unnormalized vectors $|v_k\rangle$ for all $k$ and $\sum_k p_k =1$, then there exists a pure state decomposition $\{p_k,\ |\psi_k\rangle\}$ with $|\psi_k\rangle=\Delta(\rho)^{\frac{1}{2}} |v_k\rangle$ and $\rho=\sum_k p_k  |\psi_k\rangle\langle\psi_k| $. Furthermore $\Delta(|\psi_k\rangle\langle\psi_k|)=\Delta(\rho)^{\frac{1}{2}} \Delta(|v_k\rangle\langle v_k|) \Delta(\rho)^{\frac{1}{2}} =\Delta(\rho)$. By Theorem \ref{th upper bound}, $\{p_k,\ |\psi_k\rangle\}$ is a pure state decomposition of $\rho$ such that ${\bar C}$ reaches the upper bound $f(\rho_{11}, \rho_{22}, \cdots, \rho_{nn})$.
\end{proof}

Since the set of correlation matrices is convex and its extreme points are rank one in two and three dimensional systems \cite{CV}, therefore all correlation matrices $CM(\rho)$ related to the quantum state $\rho$ in two and three dimensional systems can be decomposed as the convex combination of rank one correlation matrices. Combining this fact with Theorem \ref{th saturation}, we obtain the following result.

\begin{corollary}
In two and three dimensional systems, the maximum of ${\bar C}$ is $f(\rho_{11}, \rho_{22}, \cdots, \rho_{nn})$ for any quantum state $\rho$.
\end{corollary}

Now we show the optimal pure state decompositions of mixed state reaching the upper bound in Theorem \ref{th upper bound} by two explicit examples.

{\bf Example 1.}
For any qubit state $\rho=\sum_{i,j=0}^1 \rho_{ij} |i\rangle \langle j|$, the corresponding correlation matrix is
$CM(\rho)=\left(
\begin{array}{ccccccc}
1 & \rho_{12}/\sqrt{\rho_{11}\rho_{22}}\\
\rho_{21}/\sqrt{\rho_{11}\rho_{22}} & 1
\end{array}
\right)$, which can be decomposed as $CM(\rho)=p_1|\phi_1\rangle\langle\phi_1| +p_2|\phi_2\rangle\langle\phi_2|$ with $|\phi_1\rangle=|0\rangle+e^{{\rm -i arg}({\rho_{12}})}|1\rangle$, $p_1=\frac{1}{2}(1+ |\rho_{12}|/\sqrt{\rho_{11}\rho_{22}})$ and $|\phi_2\rangle=|0\rangle-e^{{\rm -i arg}({\rho_{12}})} |1\rangle$, $p_2=\frac{1}{2}(1- |\rho_{12}|/\sqrt{\rho_{11}\rho_{22}})$, for nonzero $\rho_{11}$ and $\rho_{22}$, ${\rm arg}(\rho_{12} )$ is the argument of $\rho_{12}$,
and the roman i denotes the imaginary unit.
Then $\rho$ can be decomposed as the pure state decomposition $\mathfrak{D}=\{p_k,\ |\psi_k\rangle\}_{k=0,1}$ whose average coherence attains the maximum, ${\bar C}(\mathfrak{D})=f(\rho_{11},\ \rho_{22})$, where
$|\psi_1\rangle=\sqrt{\rho_{11}}|0\rangle+e^{{\rm- i arg}({\rho_{12}})} \sqrt{\rho_{22}} |1\rangle$, $|\psi_2\rangle=\sqrt{\rho_{11}}|0\rangle-e^{{\rm- i arg}({\rho_{12}})}  \sqrt{\rho_{22}} |1\rangle$, with the same weights $p_1$ and $p_2$ as $CM(\rho)$.

{\bf Example 2.}
In high dimensional systems, for the incoherent state $\rho=\sum_{i=0}^{n-1} \rho_{ii} |i\rangle\langle i|$, since $CM(\rho)=I=\frac{1}{n}\sum_k |\phi_k\rangle \langle \phi_k|$ with $|\phi_k\rangle= \sum_{j=0}^{n-1} e^{2\pi {\rm i} (k-1)j/n}|j\rangle$ for $k=1,2,\cdots,n$, so $\mathfrak{D}=\{p_k,\ |\psi_k\rangle\}$ with $|\psi_k\rangle=\sum_{j=0}^{n-1} e^{2\pi {\rm i} (k-1)j/n}\sqrt{\rho_{jj}}|j\rangle$, and $p_k=\frac{1}{n}$ for $k=0,1,\cdots, n-1$ is an optimal pure state decomposition of $\rho$ reaching the maximal average coherence, ${\bar C}(\mathfrak{D})=f(\rho_{11}, \rho_{22}, \cdots, \rho_{nn})$.

Especially, if the function $f$ is strictly concave, then one can check the proof of Theorem \ref{th upper bound} and get the equality holds if and only if there exists a pure state decomposition such that all the coherence vectors are the same for all pure states. This makes
the condition in Theorem \ref{th upper bound} necessary and sufficient for strictly concave function. Because the conditions in Theorems \ref{th upper bound} and \ref{th saturation} are equivalent, the condition in Theorem \ref{th saturation} becomes also necessary and sufficient for strictly concave function.
In fact, most of the symmetric concave functions related to the coherence measures in the literature are strictly concave.
For example,
for the relative entropy of coherence $C_r(\rho)$ \cite{T. Baumgratz} and intrinsic randomness of coherence $C_R(\rho)$ \cite{X. Yuan}, the corresponding symmetric concave function is $f(|\psi_0|^2, |\psi_1|^2,\cdots, |\psi_{n-1}|^2)=H (|\psi_0|^2, |\psi_1|^2,\cdots, |\psi_{n-1}|^2)$ where $H$ is the entropy function. For the $l_1$ norm of coherence $C_{l_1}$ \cite{T. Baumgratz} and coherence concurrence $C_{con}$ \cite{X. Qi}, the corresponding symmetric concave function is $f(|\psi_0|^2, |\psi_1|^2,\cdots, |\psi_{n-1}|^2)=\sum_{i \neq j} |\psi_i \psi_j|$. For the fidelity based measure of coherence $C_F(\rho)$ \cite{C. L. Liu2017,Streltsov2010}, the corresponding symmetric concave function is $f(|\psi_0|^2, |\psi_1|^2,\cdots, |\psi_{n-1}|^2)=\sqrt{1-\max \{|\psi_0|^2, |\psi_1|^2,\cdots, |\psi_{n-1}|^2\}}$.
These functions are all strictly concave.
In this case, the average coherence reaches its upper bound $ f(\rho_{11}, \rho_{22}, \cdots, \rho_{nn})$ if and only if
there is a pure state decomposition $\{p_k,\  |\psi_k\rangle\}$ of $\rho$ such that $\Delta(|\psi_k\rangle\langle\psi_k|)=\Delta(\rho)$ for all $k$.

Now we evaluate the distance between the average coherence and its maximum by explicit examples. Here we consider the $l_1$ norm of coherence $C_{l_1}$ and the quantum states in the form of $\rho(x)=1/2|0\rangle\langle 0| + x|0\rangle\langle 1| +x|1\rangle\langle 0| + 1/2|1\rangle\langle 1|$,
with real variable $0\leq x\leq 1/2$.
The eigenvectors of $\rho(x)$ are $|+\rangle=\frac{1}{\sqrt{2}}(|0\rangle + |1\rangle)$ and $|-\rangle=\frac{1}{\sqrt{2}}(|0\rangle - |1\rangle)$ which are independent of the variable $x$.
Suppose $\lambda_1(x)$ and $\lambda_2(x)$ are the eigenvalues of $\rho(x)$ corresponding to the eigenvectors $|+\rangle$ and $|-\rangle$,
so we have $\rho(x)=\lambda_1(x) |+\rangle \langle+| + \lambda_2(x) |-\rangle \langle-|$.
Let $\sqrt{p_1(x)} |\psi_1(x)\rangle=\cos \alpha \sqrt{\lambda_1(x)} |+\rangle + \sin \alpha \sqrt{\lambda_1(x)} |-\rangle $, $\sqrt{p_2(x)} |\psi_2(x)\rangle= -\sin \alpha \sqrt{\lambda_1(x)} |+\rangle + \cos \alpha \sqrt{\lambda_2(x)} |-\rangle $, with probability ${p_1(x)}= \cos^2 \alpha \lambda_1(x) + \sin^2 \alpha \lambda_2(x)$  and probability ${p_2(x)}= \sin^2 \alpha \lambda_1(x) + \cos^2 \alpha \lambda_2(x)$, then $\mathfrak{D}(x)=\{p_k(x), |\psi_k(x)\rangle\}_{k=1,2}$ is any pure state decomposition of $\rho(x)$ with two components. The average $l_1$ norm of coherence of $\rho(x)$ with respect to the pure state decomposition $\mathfrak{D}(x)$ is ${\bar C}_{l_1} (\mathfrak{D}(x))=\sum_k  p_k(x) C_{l_1}(|\psi_k(x) \rangle)$.
See FIG. 1 for the average coherence and MSD of $\rho(x)$ with respect to the pure state decomposition $\mathfrak{D}(x)$ with $\alpha$ in the interval $[0, \frac{\pi}{2}]$. Now we consider the interval $[0, \frac{\pi}{4}]$ for $\alpha$ due to the symmetry.
The maximum average coherence of $\rho(x)$ is 1, which is arrived at $\alpha=0$. The minimum average coherence of $\rho(x)$ is $2x$ for $\alpha\in[\frac{1}{2} \arccos{2x}, \frac{\pi}{4}]$.
The average coherence of $\rho(x_1)$ and $\rho(x_2)$ overlap on the interval $[0, \frac{1}{2} \arccos{2x_2}]$ for $0 \leq x_1\leq x_2 \leq \frac{\pi}{4}$.
The mean square deviation (MSD) of $\rho(x)$ is $MSD(\mathfrak{D}(x))=\sum_k p_k(x) (C_{l_1}(|\psi_k(x)\rangle)-{\bar C}_{l_1}(\mathfrak{D}(x)))^2$. The MSD reaches its maximum at $\alpha=\frac{1}{2} \arccos{2x}$ and minimum at $\alpha=0$ and $\alpha=\frac{\pi}{4}$.
The average coherence decreases with $\alpha$ in the interval $[0, \frac{\pi}{4}]$, and the MSD which is the distance between the pure state coherence and its average coherence increases with $\alpha$ in the interval $[0, \frac{1}{2} \arccos{2x}]$ and decreases with $\alpha$ in the interval $[\frac{1}{2} \arccos{2x}, \frac{\pi}{4}]$.

\begin{figure}[htbp]\label{fig var}
\includegraphics{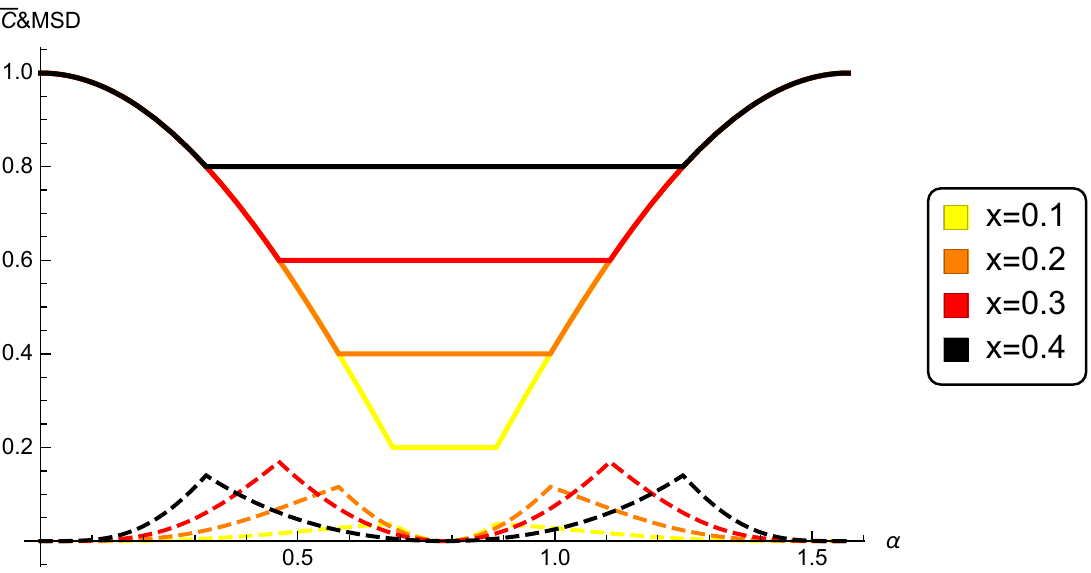}
\caption{(Color Online) The average coherence (solid lines) and MSD (dashed lines) with respect to the pure state decomposition $\mathfrak{D}(x)$ for $\rho(x)$ with $x=0.1,\ 0.2,\ 0.3,\ 0.4$. The plots are symmetric at $\alpha=\frac{\pi}{4}$ in the interval $[0, \frac{\pi}{2}]$.
In the interval $[0, \frac{\pi}{4}]$, the maximum average coherence of $\rho(x)$ is 1, which is arrived at $\alpha=0$. The minimum average coherence of $\rho(x)$ is $2x$ for $\alpha\in[\frac{1}{2} \arccos{2x}, \frac{\pi}{4}]$.
The average coherence of $\rho(x_1)$ and $\rho(x_2)$ overlap on the interval $[0, \frac{1}{2} \arccos{2x_2}]$ for $0 \leq x_1\leq x_2 \leq \frac{\pi}{4}$.
The MSD reaches its maximum at $\alpha=\frac{1}{2} \arccos{2x}$ and minimum at $\alpha=0$ and $\alpha=\frac{\pi}{4}$.}
\end{figure}

\section{The geometrical description of the pure state decompositions in qubit system}

In this section, we study the geometrical description of the pure state decompositions based on the average quantum coherence in qubit system.
On the Bloch sphere, any pure state $|\psi\rangle=\cos\frac{\theta}{2}|0\rangle+e^{{\rm i} \phi}\sin\frac{\theta}{2}|1\rangle$ is associated with a unit Bloch vector $\vec{n}(|\psi\rangle)=(\sin\theta\cos\phi, \sin\theta\sin\phi,\cos\theta)$, $0\leq \theta \leq \pi$, $0\leq \phi \leq 2\pi$. $\theta$ is the angle between the Bloch vector $\vec{n}(|\psi\rangle)$ and the axis $\langle \sigma_3 \rangle$. $\phi$ is the angle between the Bloch vector $\vec{n}(|\psi\rangle)$ and the axis $\langle \sigma_1 \rangle$. The coherence of $|\psi\rangle$ is  $C(|\psi\rangle)=f(|\cos\frac{\theta}{2}|^2, |\sin\frac{\theta}{2}|^2)=f(|\cos\frac{\theta}{2}|^2, 1-|\cos\frac{\theta}{2}|^2)$ by Eq. (\ref{cf}) which only depends on the angle $\theta$. Any mixed
state $\rho=\frac{1}{2}(I+\vec{n}\cdot \vec{\sigma})$ is also associated with a Bloch vector $\vec{n}=n(\sin\theta\cos\phi, \sin\theta\sin\phi,\cos\theta)$ with the length of the Bloch vector $n\leq 1$, where $\vec{\sigma}=(\sigma_1,\sigma_2, \sigma_3)$ with the Pauli matrices $\sigma_1=\left(\begin{array}{cccccccc}
0 & 1 \\
1 & 0
\end{array}\right)$, $\sigma_2=\left(\begin{array}{cccccccc}
0 & -i \\
i & 0
\end{array}\right)$, $\sigma_3=\left(\begin{array}{cccccccc}
1 & 0 \\
0 & -1
\end{array}\right)$.
Next we study and compare the average coherence with respect to four different pure state decompositions $\mathfrak{D}^{(M)}$, $\mathfrak{D}^{(s)}$, $\mathfrak{D}^{(m_1)}$, and $\mathfrak{D}^{(m_2)}$ of the mixed state $\rho$.
Without loss of generality, we assume $0\leq \theta \leq \pi/2$ which means the Bloch vector $\vec{n}$ is in the first octant of the Bloch sphere.

\begin{figure}[htbp]
\centering

\subfigure[The pure state decomposition $\mathfrak{D}^{(M)}$.]
{
\begin{minipage}[t]{0.45\textwidth}
\includegraphics[width=0.9\textwidth,height=0.8\textwidth]{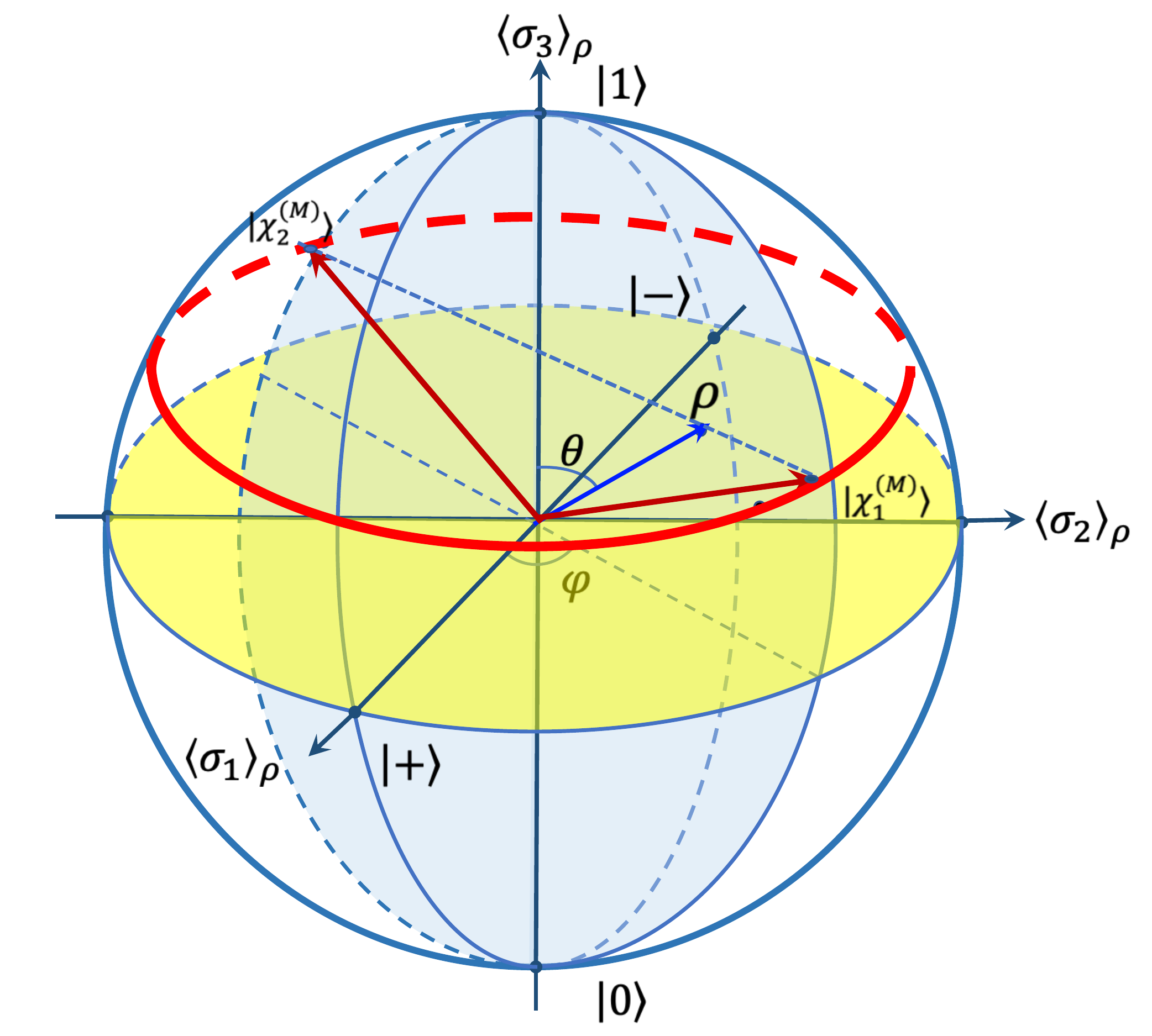}
\end{minipage}
}
\quad
\subfigure[The spectral decomposition $\mathfrak{D}^{(s)}$.]
{
\begin{minipage}[t]{0.45\textwidth}
\includegraphics[width=0.9\textwidth,height=0.8\textwidth]{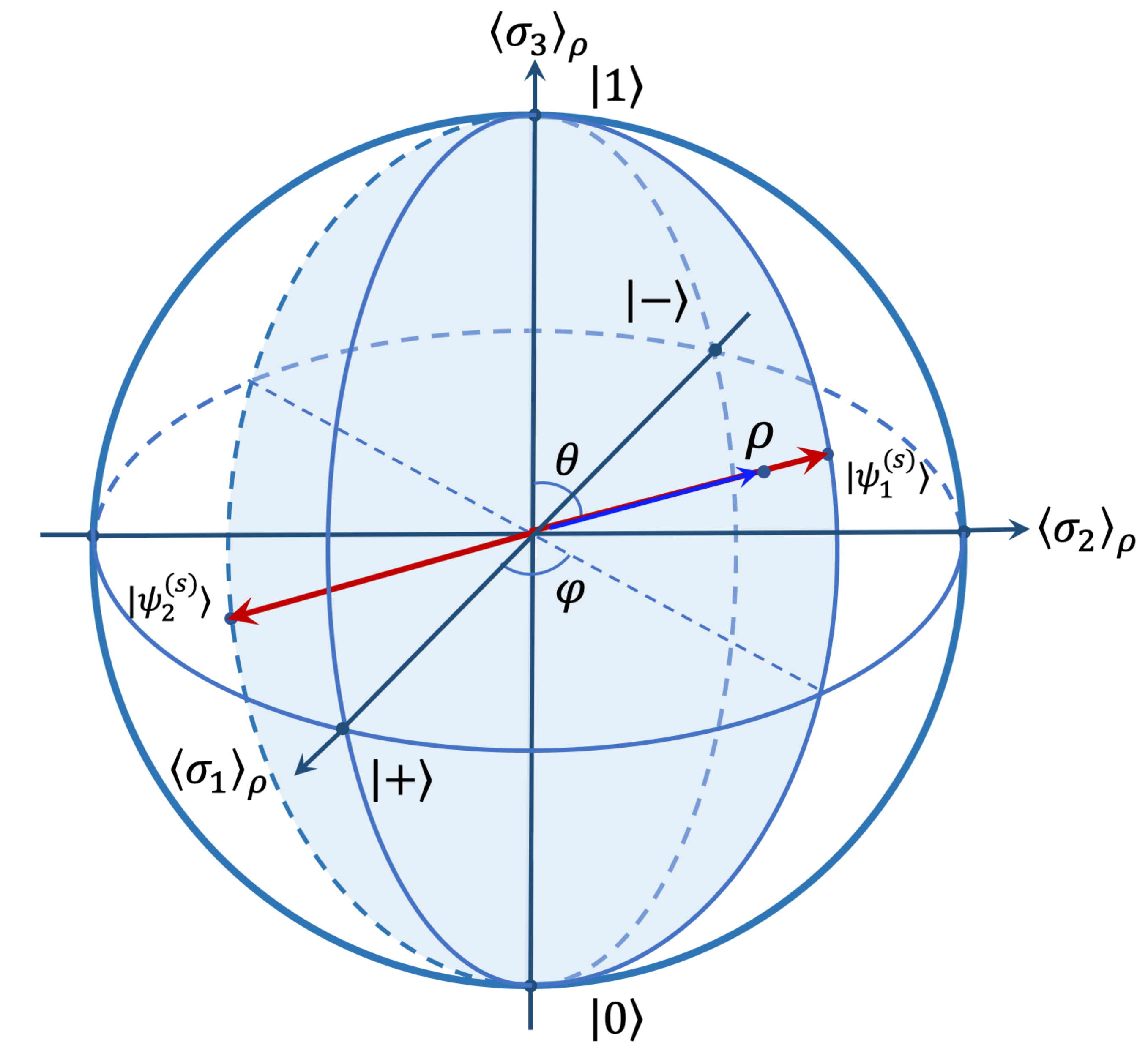}
\end{minipage}
}

\subfigure[The pure state decomposition $\mathfrak{D}^{(m_1)}$.]
{
\begin{minipage}[t]{0.45\textwidth}
\centering
\includegraphics[width=0.9\textwidth,height=0.8\textwidth]{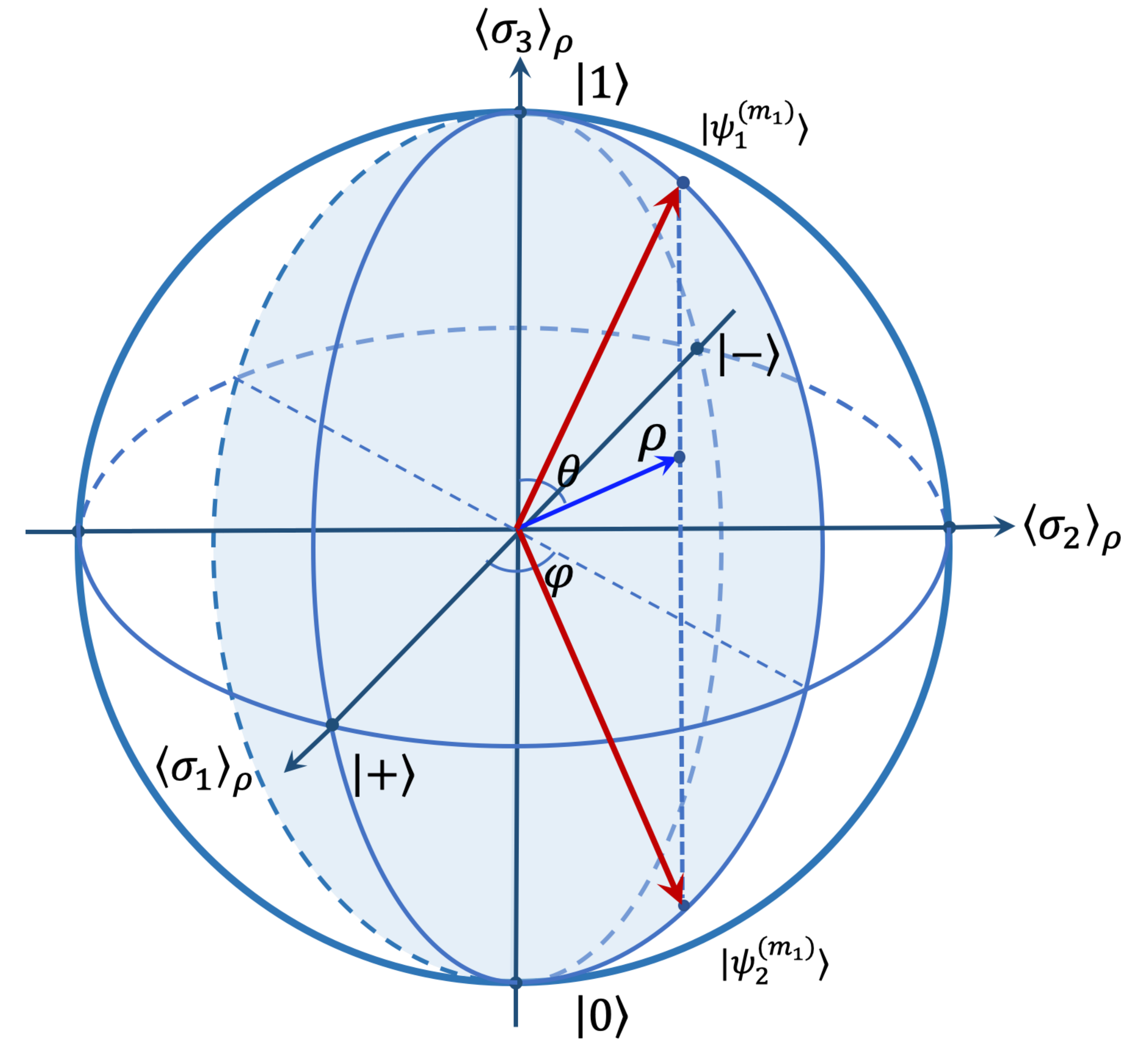}
\end{minipage}
}
\quad
\subfigure[The pure state decomposition $\mathfrak{D}^{(m_2)}$.]
{
\begin{minipage}[t]{0.45\textwidth}
\centering
\includegraphics[width=0.9\textwidth,,height=0.8\textwidth]{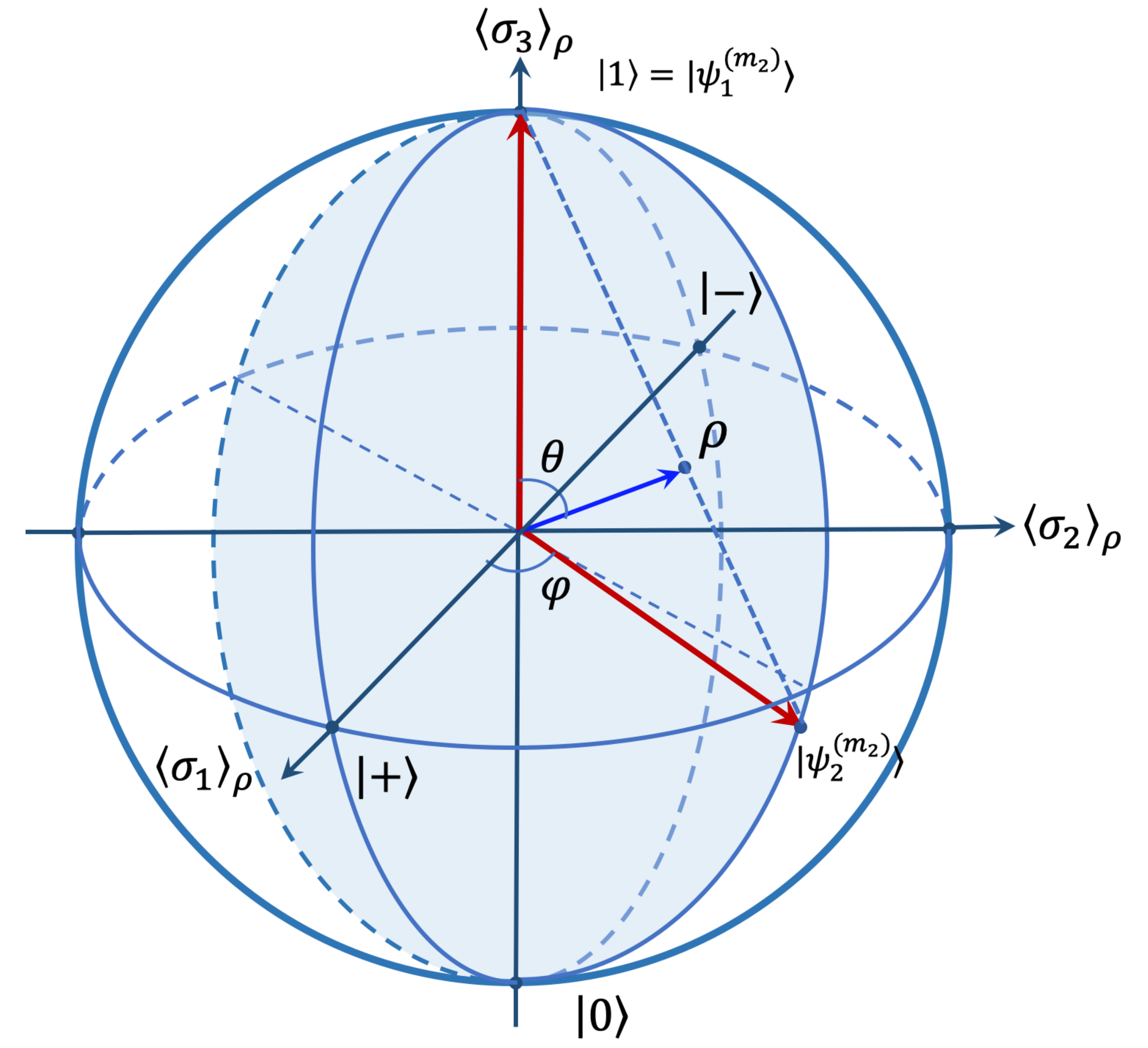}
\end{minipage}
}
\centering
\caption{The four pure state decompositions $\mathfrak{D}^{(M)}$, $\mathfrak{D}^{(s)}$, $\mathfrak{D}^{(m_1)}$, $\mathfrak{D}^{(m_2)}$ in the Bloch sphere. The central (yellow) disc in (a) is the set of mixed states with spectral decompositions reaching the maximal average coherence.}
\end{figure}

(1) {\it The pure state decomposition $\mathfrak{D}^{(M)}=\{p_k^{(M)}, |\chi_k^{(M)}\rangle\}_{k=1}^2$ reaching the maximal average coherence.} Here $|\chi_1^{(M)}\rangle=\cos\frac{\theta_o}{2}|0\rangle+e^{i\phi}\sin\frac{\theta_o}{2}|1\rangle$ with probability $p_1^{(M)}=\frac{1}{2}(1+\frac{n\sin\theta}{\sin\theta_o})$, $|\chi_2^{(M)}\rangle=\cos\frac{\theta_o}{2}|0\rangle-e^{i\phi}\sin\frac{\theta_o}{2}|1\rangle$ with probability $p_2^{(M)}=\frac{1}{2}(1-\frac{n\sin\theta}{\sin\theta_o})$, of which the Bloch vectors are $\vec{n}(|\chi_1^{(M)}\rangle)=(\sin\theta_o\cos\phi, \sin\theta_o\sin\phi,\cos\theta_o)$, $\vec{n}(|\chi_2^{(M)}\rangle)=(\sin\theta_o\cos(\pi+\phi), \sin\theta_o\sin(\pi+\phi),\cos\theta_o)$ with $\theta_o=\arccos{n \cos\theta}$ (See (a) in FIG. 2).
For this pure state decomposition, the average coherence is
\begin{equation}
{\bar C}(\mathfrak{D}^{(M)})=f(|\cos\frac{\theta_o}{2}|^2, |\sin\frac{\theta_o}{2}|^2).
\end{equation}
This average coherence is at its maximum because this pure state decomposition satisfies the sufficient condition in Theorem \ref{th upper bound}.
The pure state decompositions in the circle $\theta=\theta_o$ (The bold circle in (a) in FIG. 2) on the surface of the Bloch sphere are all optimal and are the only optimal pure state decompositions for strictly concave function $f$. In the Bloch sphere, we can see the maximal average coherence of $\rho$ is only determined by the length $n$ of the Bloch vector $\vec{n}$ and the angle $\theta_o$ with the axis $\langle \sigma_3\rangle$.

(2) {\it The spectral decomposition $\mathfrak{D}^{(s)}=\{p_k^{(s)}, |\psi_k^{(s)}\rangle\}_{k=1}^2$}. Here $|\psi_1^{(s)}\rangle=\cos\frac{\theta}{2}|0\rangle+e^{i\phi}\sin\frac{\theta}{2}|1\rangle$ with probability $p_1^{(s)}=\frac{1+n}{2}$, $|\psi_2^{(s)}\rangle=\sin\frac{\theta}{2}|0\rangle-e^{i\phi}\cos\frac{\theta}{2}|1\rangle$ with probability $p_2^{(s)}=\frac{1-n}{2}$, of which the Bloch vectors are $\vec{n}(|\psi_1^{(s)}\rangle)=(\sin\theta\cos\phi, \sin\theta\sin\phi,\cos\theta)$ and $\vec{n}(|\psi_2^{(s)}\rangle)=(\sin(\pi-\theta)\cos(\pi+\phi), \sin(\pi-\theta)\sin(\pi+\phi),\cos(\pi-\theta))=-\vec{n}(|\psi_1^{(s)}\rangle)$ (See (b) in FIG. 2).
For the spectral decomposition the average coherence is
\begin{eqnarray}
{\bar C}(\mathfrak{D}^{(s)})&=&p_1^{(s)} f(|\cos\frac{\theta}{2}|^2, |\sin\frac{\theta}{2}|^2) + p_2^{(s)} f(|\sin\frac{\theta}{2}|^2, |\cos\frac{\theta}{2}|^2)\\\nonumber
&=&f(|\cos\frac{\theta}{2}|^2, |\sin\frac{\theta}{2}|^2).
\end{eqnarray}

(3) {\it The pure state decomposition $\mathfrak{D}^{(m_1)}=\{p_k^{(m_1)}, |\psi_k^{(m_1)}\rangle\}_{k=1}^2$.} Here  $|\psi_1^{(m_1)}\rangle=\cos\frac{\theta_3}{2}|0\rangle+e^{i\phi}\sin\frac{\theta_3}{2}|1\rangle$ with probability $p_1^{(m_1)}=(-1 + n^2 \sin^2{\theta} - n \cos{\theta} \sqrt{1 - n^2 \sin^2{\theta}})/(-2 + 2 n^2 \sin^2{\theta})$, $|\psi_2^{(m_1)}\rangle=\sin\frac{\theta_3}{2}|0\rangle+e^{i\phi}\cos\frac{\theta_3}{2}|1\rangle$ with probability $p_2^{(m_1)}=(-1 + n^2 \sin^2{\theta} + n \cos{\theta} \sqrt{1 - n^2 \sin^2{\theta}})/(-2 + 2 n^2 \sin^2{\theta})$, of which the Bloch vectors are $\vec{n}(|\psi_1^{(m_1)}\rangle)=(\sin\theta_3\cos\phi, \sin\theta_3\sin\phi,\cos\theta_3)$, $\vec{n}(|\psi_2^{(m_1)}\rangle)=(\sin(\pi-\theta_3)\cos\phi, \sin(\pi-\theta_3)\sin\phi,\cos(\pi-\theta_3))$ with $\theta_3=\arccos{\sqrt{1 - n^2 \sin^2{\theta}}}$ (See (c) in FIG. 2).
For this pure state decomposition the average coherence is
\begin{equation}
{\bar C}(\mathfrak{D}^{(m_1)})=f(|\cos\frac{\theta_3}{2}|^2, |\sin\frac{\theta_3}{2}|^2).
\end{equation}

(4) {\it The pure state decomposition $\mathfrak{D}^{(m_2)}=\{p_k^{(m_2)}, |\psi_k^{(m_2)}\rangle\}_{k=1}^2$}. Here $|\psi_1^{(m_2)}\rangle=|1\rangle$ with probability $p_1^{(m_2)}=(1- n^2)/(2 (1 + n\cos\theta))$, $|\psi_2^{(m_2)}\rangle=\cos\frac{\theta_4}{2}|0\rangle+e^{i\phi}\sin\frac{\theta_4}{2}|1\rangle$ with probability $p_2^{(m_2)}=(1 + n^2 + 2n\cos\theta)/(2 (1 + n\cos\theta))$. The Bloch vectors $\vec{n}(|\psi_1^{(m_2)}\rangle)=(0,0,-1)$, $\vec{n}(|\psi_2^{(m_2)}\rangle)=(\sin\theta_4\cos\phi, \sin\theta_4\sin\phi,\cos\theta_4)$ with $\theta_4=\arccos[(1 + 2n\cos\theta + n^2 \cos{2 \theta})/(1 + n^2 + 2n\cos\theta)]$ (See (d) in FIG. 2).
For this pure state decomposition the average coherence is
\begin{equation}
{\bar C}(\mathfrak{D}^{(m_2)})=p_2^{(m_2)} f(|\cos\frac{\theta_4}{2}|^2, |\sin\frac{\theta_4}{2}|^2).
\end{equation}

Now we compare the average coherence of these four pure state decompositions. Here we employ the concepts of majorization relation and Schur concavity. Let $\textbf{x}=(x_1,\ x_2, \cdots,\ x_n)$ and $\textbf{y}=(y_1,\ y_2, \cdots,\ y_n)$ be two vectors in the probability simplex $\Omega$ with coordinates in the decreasing order. $\textbf{x}$ is majorized by $\textbf{y}$ denoted as $\textbf{x} \prec \textbf{y}$, if $\sum_{i=1}^k x_k \leq \sum_{i=1}^k y_k$ for $1 \leq k \leq n$. For any real valued function $f$, it is called Schur concave if $\textbf{x} \prec \textbf{y}\Rightarrow f(\textbf{x}) \geq f(\textbf{y})$. Factually every real symmetric concave function is Schur concave \cite{R. Bhatia}.

By the first two pure state decompositions we can see
${\bar C}(\mathfrak{D}^{(s)})\leq {\bar C}(\mathfrak{D}^{(M)})$.
The inequality becomes equality if and only if $\theta=\frac{\pi}{2}$. So
the spectral decomposition reaches the maximal average coherence if and only if the mixed state $\rho$ lies in the central disc of Bloch sphere (See (a) in FIG. 2).

On one hand, since $|\cos\frac{\theta_3}{2}|^2>|\cos\frac{\theta}{2}|^2$, we get $(|\cos\frac{\theta}{2}|^2, |\sin\frac{\theta}{2}|^2)\prec (|\cos\frac{\theta_3}{2}|^2, |\sin\frac{\theta_3}{2}|^2)$, which means
$f(|\cos\frac{\theta_3}{2}|^2, |\sin\frac{\theta_3}{2}|^2)\leq f(|\cos\frac{\theta}{2}|^2, |\sin\frac{\theta}{2}|^2)$ because $f$ is Schur concave.
Therefore the order relations among the average coherence for the pure state decompositions $\mathfrak{D}^{(m_1)}$, $\mathfrak{D}^{(s)}$ and $\mathfrak{D}^{(M)}$ is
\begin{equation}\label{eq order 2}
{\bar C}(\mathfrak{D}^{(m_1)}) \leq {\bar C}(\mathfrak{D}^{(s)})\leq {\bar C}(\mathfrak{D}^{(M)}).
\end{equation}

On the other hand, since $|\cos\frac{\theta_4}{2}|^2>|\cos\frac{\theta}{2}|^2$, we get $(|\cos\frac{\theta}{2}|^2, |\sin\frac{\theta}{2}|^2)\prec (|\cos\frac{\theta_4}{2}|^2, |\sin\frac{\theta_4}{2}|^2)$ and
$f(|\cos\frac{\theta_4}{2}|^2, |\sin\frac{\theta_4}{2}|^2)\leq f(|\cos\frac{\theta}{2}|^2, |\sin\frac{\theta}{2}|^2)$.
Therefore the order relations among the average coherence for the pure state decompositions $\mathfrak{D}^{(m_2)}$, $\mathfrak{D}^{(s)}$ and $\mathfrak{D}^{(M)}$ is
\begin{equation}\label{eq order 1}
{\bar C}(\mathfrak{D}^{(m_2)}) < {\bar C}(\mathfrak{D}^{(s)})\leq {\bar C}(\mathfrak{D}^{(M)}).
\end{equation}

The order relations (\ref{eq order 2}) and (\ref{eq order 1}) are invariant for all coherence measures. This result is peculiar because the order relation of the average coherence with respect to different pure state decompositions probably change for different coherence measures. Typically for any mixed state $\rho$, its average coherence with respect to pure state decomposition $\mathfrak{D}=\{p_k, |\psi_k\rangle\}_{k=1}^d$ with $|\psi_k\rangle=\cos\frac{\theta_k}{2}|0\rangle+e^{{\rm i} \phi_k}\sin\frac{\theta_k}{2}|1\rangle$ is ${\bar C}(\mathfrak{D})=\sum_{k=1}^d p_k f(|\cos\frac{\theta_k}{2}|^2, 1-|\cos\frac{\theta_k}{2}|^2)$, which a multi-variable function. So it is not easy to compare two average coherence, not to mention the invariance of their order relations.
The invariance of the order relations (\ref{eq order 2}) and (\ref{eq order 1}) is attributed to the symmetry of these pure state decompositions. Because of the symmetry of the distribution of Bloch vectors in the sphere, the average coherence of $\mathfrak{D}^{(M)}$, $\mathfrak{D}^{(s)}$, $\mathfrak{D}^{(m_1)}$ turn to a single variable function. Because of the incoherent component in the pure state decomposition $\mathfrak{D}^{(m_2)}$, the average coherence ${\bar C} (\mathfrak{D}^{(m_2)})$ is bounded from above by the same single variable function. So this makes the comparison of these four average coherence possible. Finally thanks to the concavity of function $f$, we find the order relations (\ref{eq order 2}) and (\ref{eq order 1}) are invariant for all coherence measures.

For any given mixed state $\rho$, the coherence is not only embodied in the off diagonal entries of its density matrix, but also hides in its pure state decompositions. Intuitively, there may be order relations among the average coherence of the pure state decompositions for one coherence measure, but there is the possibility that is probably be different for another coherence measure. The invariance of the order relations of the average coherence as in (\ref{eq order 2}) and (\ref{eq order 1}) for all coherence measures is surprising. As an application, it is helpful to the assisted coherence distillation, in which one needs to know the pure state decomposition with average coherence as much as possible in order to enhance the coherence assisted by another party \cite{E. Chitambar,zhao l1max}. The invariance of the order relations among the average coherence of the pure state decompositions provides the possible pure state decompositions which is independent of explicit coherence measures. This universality is one of the intrinsic features of coherence.

For the above four pure state decompositions $\mathfrak{D}^{(M)}$, $\mathfrak{D}^{(s)}$, $\mathfrak{D}^{(m_1)}$ and $\mathfrak{D}^{(m_2)}$, and for the coherence measures such as
the relative entropy of coherence $C_r(\rho)$ or intrinsic randomness of coherence $R(\rho)$, the $l_1$ norm of coherence $C_{l_1}$ or coherence concurrence $C_{con}$ and the fidelity based measure of coherence $C_F(\rho)$, the order relation of the average coherence is
\begin{equation}
{\bar C}(\mathfrak{D}^{(m_1)})\leq {\bar C}(\mathfrak{D}^{(m_2)}) < {\bar C}(\mathfrak{D}^{(s)})\leq {\bar C}(\mathfrak{D}^{(M)}).
\end{equation}
The explicit average coherence are summarized in Table I.
Furthermore, the pure state decomposition $\mathfrak{D}^{(M)}$ is optimal for reaching the maximal average coherence and $\mathfrak{D}^{(m_1)}$ is optimal for reaching the minimal average coherence for these coherence measures. By the fact that the
pure state decomposition $\mathfrak{D}^{(M)}$ is optimal for the maximal average coherence for all coherence measures, one may conjecture the pure state decomposition $\mathfrak{D}^{(m_1)}$ is optimal for the minimal average coherence for all coherence measures. However, the answer is negative. Let us consider the function $f$ defined on the two dimensional probability simplex as $f(|\psi_0|^2, |\psi_1|^2)= 1 - (2 (|\psi_0|^2 - 1/2))^{2k}$ for $|\psi_0|^2\geq 1/2$ and $f(|\psi_0|^2, |\psi_1|^2)=f(|\psi_1|^2, |\psi_0|^2)$ for $|\psi_0|^2< 1/2$. One can verify that $f$ is a real symmetric concave function. So the corresponding $C_f$ as defined in Eq. (\ref{cf}) is a coherence measure for pure states.
When $k$ is large enough, the average coherence ${\bar C}_f(\mathfrak{D}^{(m_1)})$ is strictly larger than the average coherence ${\bar C}_f(\mathfrak{D}^{(m_2)})$. Therefore the pure state decomposition which achieves the minimal average coherence  for any fixed density matrix is not unique for all coherence measures.

\begin{table}[!h]
\newcommand{\tabincell}[2]{\begin{tabular}{@{}#1@{}}#2\end{tabular}}
\begin{center}
\def\temptablewidth{1\textwidth}
{\rule{\temptablewidth}{2pt}}
\begin{tabular*}
{\temptablewidth}{@{\extracolsep{\fill}}c|c|c|c}
& $C_r$ ($C_R$) & $C_{l_1}$ ($C_{con}$) &  $C_F$   \\   \hline
$f$ & $H (|\psi_0|^2, |\psi_1|^2)$ & $\sum_{i \neq j} |\psi_i \psi_j|$ & $\sqrt{1-\max \{|\psi_0|^2, |\psi_1|^2\}}$  \\   \hline
$\mathfrak{D}^{(M)}$ & $H(|\cos \frac{\theta_o}{2}|^2, |\sin \frac{\theta_o}{2}|^2)$  &  $|\cos \frac{\theta_o}{2} \sin \frac{\theta_o}{2}|$ &  $\sqrt{ 1-\max\{|\cos \frac{\theta_o}{2}|^2, |\sin \frac{\theta_o}{2}|^2\}}$  \\  \hline
$\mathfrak{D}^{(s)}$ & $H(|\cos \frac{\theta}{2}|^2, |\sin \frac{\theta}{2}|^2)$  &  $|\cos \frac{\theta}{2} \sin \frac{\theta}{2}|$ & $\sqrt{1-\max\{|\cos \frac{\theta}{2}|^2, |\sin \frac{\theta}{2}|^2\}}$  \\  \hline
$\mathfrak{D}^{(m_1)}$ & $H(|\cos \frac{\theta_3}{2}|^2, |\sin \frac{\theta_3}{2}|^2)$  &  $|\cos \frac{\theta_3}{2} \sin \frac{\theta_3}{2}|$ & $\sqrt{ 1-\max\{|\cos \frac{\theta_3}{2}|^2, |\sin \frac{\theta_3}{2}|^2\}}$  \\   \hline
$\mathfrak{D}^{(m_2)}$ & $ p_2^{(m_2)} H(|\cos \frac{\theta_4}{2}|^2, |\sin \frac{\theta_4}{2}|^2)$ &  $ p_2^{(m_2)} |\cos \frac{\theta_4}{2} \sin \frac{\theta_4}{2}|$ &  $ p_2^{(m_2)} \sqrt{1-\max\{|\cos \frac{\theta_4}{2}|^2, |\sin \frac{\theta_4}{2}|^2\}}$
\end{tabular*}
{\rule{\temptablewidth}{2pt}}
 \caption{The average coherence of four pure state decompositions $\mathfrak{D}^{(M)}$, $\mathfrak{D}^{(s)}$, $\mathfrak{D}^{(m_1)}$ and $\mathfrak{D}^{(m_2)}$ with respect to the relative entropy of coherence $C_r$ or intrinsic randomness of coherence $C_R$, the $l_1$ norm of coherence $C_{l_1}$ or coherence concurrence $C_{con}$ and the fidelity based measure of coherence $C_F$.
  }
 \end{center}
 \end{table}

\section{Conclusions}

In conclusion, we have studied the average quantum coherence over the pure state decompositions of a quantum state. We obtain an upper bound of the average quantum coherence solely as a function of the diagonal entries of the quantum state.
Hence, one can estimate the average quantum coherence by measuring only the diagonal entries, instead of getting the full ensemble information.
Two sufficient conditions for the saturation of the upper bound is shown which always hold for two and three dimensional systems. The pure state decomposition reaching the maximal average coherence is independent of any explicit coherence measure.
Our second result exhibits an interesting geometrical description of the pure state decompositions of a qubit state in Bloch sphere. For any given qubit state, the optimal pure state decomposition reaching the maximal average quantum coherence, the spectral decomposition,  as well as other two pure state decompositions are shown in the Bloch sphere. The order relations among their average quantum coherence are invariant for all coherence measures. We hope these universal results  can strengthen the understanding of quantum coherence.

\bigskip
\noindent{\bf Acknowledgments}\, Ming-Jing Zhao was supported by the China Scholarship
Council (Grant No. 201808110022), Qin Xin Talents Cultivation Program, Beijing Information Science and Technology University. Teng Ma is supported by the NSF of China (Grant No.
11905100).  Rajesh Pereira is supposrted by NSERC discovery grant no 400550.

\end{document}